\documentclass[10pt]{article}

\usepackage[english]{babel}

\usepackage[a4paper,top=2cm,bottom=2cm,left=3cm,right=3cm,marginparwidth=1.75cm]{geometry}

\usepackage[utf8]{inputenc}
\usepackage{amsmath}
\usepackage{paralist}
\usepackage{amssymb,amsthm, amsfonts}
\usepackage{mathtools}
\usepackage{enumerate}
\usepackage{subcaption}
\usepackage{graphicx}
\usepackage{comment}
\usepackage{tabularx}
\usepackage{booktabs}
\usepackage{makecell}
\usepackage{caption}
\usepackage{lipsum}
\usepackage{parskip}
\usepackage{floatrow}
\usepackage{bm}
\usepackage[misc]{ifsym}
\usepackage{epsfig} 
\usepackage{epstopdf} 
\usepackage[colorlinks=true]{hyperref}
\hypersetup{urlcolor=blue, citecolor=red}
\allowdisplaybreaks

\newtheorem{theorem}{Theorem}[section]

\newtheorem{lemma}[theorem]{Lemma}

\theoremstyle{definition}
\newtheorem{definition}[theorem]{Definition}
\newtheorem{remark}[theorem]{Remark}

\newtheorem{example}{Example}

\title{Criteria for the construction of MDS convolutional codes with good column distances}
\author{Zita Abreu$^{1,2}$, Julia Lieb$^2$, Raquel Pinto$^1$, Joachim Rosenthal$^2$\\
\vspace{-4mm}\\
University of Aveiro$^1$, University of Zurich$^2$}
\date{}

\begin{document}

\maketitle

\begin{abstract}
Maximum-distance separable (MDS) convolutional codes are characterized by the property that their free distance reaches the generalized Singleton bound. In this paper, new criteria to construct MDS convolutional codes are presented. Additionally, the obtained convolutional codes have optimal first (reverse) column distances and the criteria allow to relate the construction of MDS convolutional codes to the construction of reverse superregular Toeplitz matrices. Moreover, we present some construction examples for small code parameters over small finite fields.
\end{abstract}

\section{Introduction}
Communication systems that work with digitally represented data use error correction codes because all real communication channels are noisy. One type of error-correcting codes are convolutional codes. 
The distance of a convolutional code provides a means to assess its capability to protect data from errors. Codes with larger distance are better because they allow correcting more errors. One type of distance for convolutional codes is the free distance, which is considered for decoding when the codeword is fully received. Convolutional codes with maximal free distance are called Maximum Distance Separable Codes (MDS). \\
There are some known constructions of MDS convolutional codes. The first construction was obtained by Justesen in \cite{ju75} for codes of rate $1/n$ with restricted degrees. In \cite{sm98p1} Smarandache and Rosenthal presented constructions of convolutional codes of rate $1/n$ and arbitrary degree $\delta$. However, these constructions require a larger field size than the one obtained by Justesen. Later, Gluesing-Luerssen and Langfeld presented in \cite{gll06} a construction of convolutional codes of rate $1/n$ with the same field size as the one obtained by Justesen but also with a restriction on the degree of the code. Thereafter, Gluesing-Luerssen, Smarandache and Rosenthal \cite{sm01a} constructed MDS convolutional codes for arbitrary parameters $(n,k,\delta)$ but with the restriction that $|\mathbb F|-1$ needs to be divisible by $n$ and sufficiently large. Lieb and Pinto \cite{mdssr} presented a new way of constructing MDS convolutional codes of any degree and sufficiently low rate using superregular matrices. \\
In contrast to MDS block codes, there is no algebraic criterion to check whether a convolutional code is MDS. In \cite{mdssr} a sufficient criterion for MDS convolutional codes was presented, which however only works if the length $n$ of the code is sufficiently large with respect to rank $k$ and degree $\delta$. \\
In this paper, a new method for constructing MDS convolutional codes will be presented. This method allows the construction also for several sets of code parameters for which so far no constructions of MDS convolutional codes existed. Additionally, the codes fulfilling our new criteria have optimal first column distances and if $k\mid\delta$ also optimal first reverse column distances. In particular, we will relate these codes to reverse superregular Toeplitz matrices.
Moreover, the field size in code constructions is an important factor to be considered since it is directly related to the computational efficiency of the encoding and decoding algorithms and the complexity of the decoding algorithms. Therefore, 
we also present some construction examples for small code parameters over small finite fields.\\
The outline of the paper is as follows. In Section 2 we present some preliminaries about convolutional codes in general and about the construction of convolutional codes with optimal distances in particular. In Section 3 we present our new criteria for the construction of MDS convolutional codes with optimal first column distances. In Section 4, we show how depending on the code parameters the conditions of Section 3 can be released to enable easier construction over fields of smaller size. In Section 5.1 we present general construction methods for codes fulfilling the criteria of Section 3 and in Section 5.2 we give some concrete examples for small code parameters optimizing the field size with the help of the results of Section 4. We conclude in Section 6 with a brief summary.

\section{Preliminaries}

A \textbf{convolutional code} $\mathcal{C}$ of rate $k/n$ is an $\mathbb{F}_{q}[z]$-submodule of $\mathbb{F}_{q}^n[z]$ of rank $k$, where $\mathbb{F}_{q}[z]$ is the ring of polynomials with coefficients in the field $\mathbb{F}_{q}$. A matrix $G(z)\in\mathbb{F}_{q}[z]^{k\times n}$ whose rows constitute a basis of $\mathcal{C}$ is called a \textbf{generator matrix} for $\mathcal{C}$. It is a full row rank matrix such that: \begin{eqnarray}
\mathcal{C} =\ \text{Im}_{\mathbb{F}_{q}[z]} \ G(z)=  \{v(z) \in \mathbb{F}_{q}^{n}[z]: v(z) = u(z)G(z) \text{ with } u(z) \in \mathbb{F}_{q}^{k}[z]\}.\nonumber
\end{eqnarray}
The maximum degree of the full-size minors of a generator matrix of $\mathcal{C}$ is called the \textbf{degree} $\delta$ of $\mathcal{C}$. A convolutional code of rate $k/n$ and degree $\delta$ is also denoted as $(n, k, \delta)$ convolutional code.

\begin{definition}
Let $G(z)=\sum_{i=0}^{\mu}G_i z^{i}\in\mathbb F_q[z]^{k\times n}$ with $G_{\mu} \neq 0$ and $k\leq n$. For each $i$, $1\leq i\leq k$, the $i$-th \textbf{row degree} $\nu_i$ of $G(z)$ is defined as the largest degree of any entry in row $i$ of $G(z)$, in particular $\mu=\max_{i=1,\hdots,k}\nu_i$. 
\end{definition}

\begin{definition}
A matrix $G(z)\in \mathbb F_q[z]^{k\times n}$ is said to be
\textbf{row reduced}\index{row reduced} if the maximum degree of its full-size minors, $\delta$, is equal to the sum of its row degrees. In this case, $G(z)$ is called a minimal generator matrix of the convolutional code it generates. 
\end{definition}

\begin{definition}
$G(z)\in\mathbb F_q[z]^{k\times n}$ is said to have \textbf{generic row degrees}
if $\nu_1=\hdots=\nu_t=\lceil\frac{\delta}{k}\rceil=\mu$ and $\nu_{t+1}=\hdots=\nu_k=\lfloor\frac{\delta}{k}\rfloor$ for $t=\delta+k-k\lceil\frac{\delta}{k}\rceil$. If $k\nmid\delta$, we denote by $\tilde{G}_{\mu}$ the matrix consisting of the first $t$ nonzero rows of $G_{\mu}$.
\end{definition}

The \textbf{free distance} of a convolutional code measures its capability of detecting and correcting errors introduced during information transmission through a noisy channel and it is defined by $$d_{free}(\mathcal{C})= \min \{wt(v(z)) | v(z) \in \mathcal{C}, v(z) \neq 0\},$$ where $wt(v(z)) = \sum^{\deg(v(z))}_{t=0} wt (v_t)$ is the Hamming weight of $v(z) = \sum_{t=0}^{\deg (v(z))}v_{t}z^{t} \in \mathbb{F}_{q}^n[z]$ and the weight $wt(v)$ of $v \in \mathbb{F}_{q}^{n}$ is the number of nonzero components of $v$. 
In \cite{zora22153}, the authors obtained an upper bound for the free distance of an $(n,k,\delta)$ convolutional code $\mathcal{C}$ 
given by $$d_{free}(\mathcal{C}) \leq (n-k)\Big( \Big\lfloor \frac{\delta}{k} \Big\rfloor +1 \Big) + \delta +1=:S.$$
This bound is called \textbf{generalized Singleton bound}. An $(n,k,\delta)$ convolutional
code with free distance equal to this bound is called \textbf{Maximum Distance Separable (MDS)} convolutional code. 

Besides the free distance, for convolutional codes one usually also considers its column distances since these are the measure for the error-correcting capability if low decoding delay is required.

\begin{definition}
For a polynomial vector ${v}(z)=\sum_{j\in\mathbb{N}_0}{v}_jz^j \in \mathbb{F}_q^n[z]$, the \textbf{j-th truncation}\index{truncation} of ${v}(z)$ is defined as
\[
{v}_{[0,j]}(z)={v}_0+{v}_1z+\cdots+{v}_jz^j.
\]
For $j\in\mathbb N_0$, the \textbf{j-th column distance}\index{column distance} of a convolutional code $\mathcal{C}$ is defined as
\[
d_j^c(\mathcal{C}):=\min\left\{wt({v}_{[0,j]}(z))\ |\ {v}(z)\in\mathcal{C} \text{ and }{v}_0 \neq \mathbf{0}\right\}.
\]
\end{definition}
\begin{theorem}\cite{gl06} Let $\mathcal{C}$ be an $(n,k,\delta)$ convolutional code. Then,
$d_j^c (\mathcal{C}) \leq (n-k)(j + 1) + 1$.
\end{theorem}

In order to present a criterion for how to check if the column distances reach their upper bound, we first need to introduce the following definition.

\begin{definition}
For $r\in\mathbb N$, let $A\in\mathbb F_q^{r\times r}$ be a matrix with $(i,j)$-entry denoted by $a_{ij}$. Define $X=\{x_{ij} \;:\; i,j\in\{1,\ldots, r\}\}$ and let $\mathbb{F}_q[X]$ be the set of polynomials in the indeterminates $x_{ij}$. Then, define $\overline{A}\in\mathbb{F}_q[X]^{r\times r}$ as the matrix with $(i,j)$-entry equal to $$\overline{a}_{ij}=\begin{cases} 0 & \text{for}\ a_{ij}=0 \\ x_{ij} & \text{for}\ a_{ij}\neq 0 \end{cases}$$
 The determinant of $A$ is called \textbf{trivially zero}\index{trivially zero} if the  determinant of $\overline{A}$ is equal to the zero polynomial. Otherwise, this determinant is said to be \textbf{non trivially zero}. A matrix that has the property that all its non trivially zero minors (of all sizes) are nonzero, is called \textbf{superregular}.
\end{definition}

\begin{theorem}[\cite{gl06}] Let $\mathcal{C}$ be a convolutional code with generator matrix $G(z)=\sum_{i=0}^{\mu}G_iz^i\in \mathbb{F}_q[z]^{k\times n}$ with $G_0$ full rank and $0\leq j\leq L:=\left\lfloor\frac{\delta}{k}\right\rfloor+\left\lfloor\frac{\delta}{n-k}\right\rfloor$.
The following statements are equivalent:
  \begin{itemize}
  \item[(a)] $d_j^c (\mathcal{C})=(n-k)(j + 1) + 1$.
  \item[(b)]
   Every full-size minor of $G_j^c=\begin{pmatrix}
G_0 & \cdots & G_{j}\\
 & \ddots & \vdots\\
0 &  & G_0
\end{pmatrix}$ that is non trivially zero
     is nonzero.
  \end{itemize}
\label{mdpcrit}
\end{theorem}

The following lemmata will be needed to prove our main statements in the next sections.

\begin{lemma}[\cite{gl06}]
Let $\mathcal{C}$ be an $(n,k,\delta)$ convolutional code with generator matrix $G(z)$ and $G_0$ full rank. If $d_j^c(\mathcal{C})=(n-k)(j+1)+1$ for some
  $j\in\{1,\hdots,L\}$, then  $d_i^c(\mathcal{C})=(n-k)(i+1)+1$ for all
  $i\leq j$.
\label{maxim}
\end{lemma}

\begin{lemma}\label{le2}
Let $A \in \mathbb F_q^{r \times s}$ with $r\leq s$ be such that all its full-size minors are nonzero. Then, each vector which is a linear combination of the $r$ rows of $A$ has at least $s-r+1$ nonzero entries.
\end{lemma}

The last lemma was also used in \cite{mdssr} for the construction of MDS convolutional codes resulting in the following theorems.


\begin{theorem} \cite{mdssr} \label{th2}
Assume that $\delta< k$ and let $\begin{pmatrix}
    \tilde{G}_{\mu}\\ G_{\mu-1}\\ \vdots \\ G_0
\end{pmatrix}$ be superregular.
If $n\geq \delta+k-1$,
then $G(z)$ is the generator matrix of an $(n,k,\delta)$ MDS convolutional code.
\end{theorem}

\begin{theorem} \cite{mdssr} \label{th2}
Assume that $\delta\geq k$ and let  $\begin{pmatrix}
    \tilde{G}_{\mu}\\ G_{\mu-1}\\ \vdots \\ G_0
\end{pmatrix}$ be superregular. Moreover, assume that if $k\mid\delta$, all full-size minors of $[G_0\ \hdots G_{\mu}]$ are nonzero, and if $k\nmid\delta$, all full-size minors of $[G_0\ \hdots G_{\mu-1}]$ and  $\tilde{G}_{\mu}$ are nonzero.
If $n\geq 2\delta+k-\nu$,
then $G(z)$ is the generator matrix of an $(n,k,\delta)$ MDS convolutional code.
\end{theorem}

We will compare our new results with the results of the two preceding theorems later in the paper.

In the next section, we will relate MDS convolutional codes with codes that have the property that the code itself as well as the so-called reverse code, as introduced in the next definition, reach the upper bound for their $(\mu-1)$-th column distance.

\begin{definition}[\cite{hu08c}]
Let $\mathcal{C}$ be an $(n,k,\delta)$ convolutional code with generator matrix $G(z)$, which has entries $g_{ij}(z)$. Set $\overline{g}_{ij}(z):=z^{\nu_i}g_{ij}(z^{-1})$ where $\nu_i$ is the $i$-th row degree of $G(z)$. Then, the code $\overline{\mathcal{C}}$ with generator matrix $\overline{G}(z)$, which has $\overline{g}_{ij}(z)$ as entries, is called the \textbf{reverse code}\index{reverse code} to $\mathcal{C}$. We call the $j$-th column distance of $\overline{\mathcal{C}}$ the \textbf{$j$-th reverse column distance} of $\mathcal{C}$.
\end{definition}

While large column distances are important for error-correction during forward decoding with low delay, reverse column distances are important for backward decoding, see \cite{virtuphD}.

\begin{remark}
Let $G(z)=\sum_{i=0}^{\mu}G_iz^i$ and $\overline{G}(z)=\sum_{i=0}^{\mu}\overline{G}_iz^i$. If $k\mid\delta$, one has that $\overline{G}_i=G_{\mu-i}$ for $i=0,\hdots,\mu$.
\end{remark}


\section{Criteria and Construction of MDS convolutional codes}


In this section, we present theorems, which provide new criteria for the construction of MDS convolutional codes.
To this end, we have to distinguish the cases $k\mid\delta$ and $k\nmid\delta$.

\subsection{Criteria for MDS convolutional codes with $k\mid\delta$}

\begin{theorem}\label{main}
Consider an $(n,k,\delta)$ convolutional code $\mathcal{C}$ with $n>k$ and $k\mid\delta$ and minimal generator matrix $G(z)=\sum_{i=0}^{\mu}G_iz^i$ with generic row degrees, i.e. $\mu=\frac{\delta}{k}$.\\ If $\mu\geq 3$ assume $n\geq  3k- \frac{2k}{\delta - 2k}$, 
except for the case 
$k=2$ and $\delta = 6$ where we assume $n \geq 5$ and let all non trivially zero full-size minors of the following matrices be nonzero:
\begin{align*} &\begin{pmatrix}
G_0 & \cdots & G_{\mu-1}\\
 & \ddots & \vdots\\
0 &  & G_0
\end{pmatrix}\  \text{and} \
\begin{pmatrix}
G_{\mu} & \cdots & G_{1}\\
 & \ddots & \vdots\\
0 &  & G_{\mu}\end{pmatrix}\  \text{and}\\  &\begin{pmatrix}
G_{\ell} & \cdots & G_{\mu}\\
\vdots  &   & \vdots\\
G_0 & \cdots & G_{\mu-\ell}
\end{pmatrix} \  \text{for}\  0\leq \ell<min\left(\mu-1,\frac{n(\mu+1)-k+1}{n+k}\right).
\end{align*}
 If $\mu\leq 2$, let all non trivially zero full-size minors of 
$\begin{pmatrix}
G_0 & \cdots & G_{\mu}\\
 & \ddots & \vdots\\
0 &  & G_0
\end{pmatrix}$, $\begin{pmatrix}
G_{\mu} & \cdots & G_{1}\\
 & \ddots & \vdots\\
0 &  & G_{\mu}\end{pmatrix}$ and $[G_0\ \hdots \ G_{\mu}]$ be nonzero and assume for $\mu=1$ that $n \geq 2k-1$ and for $\mu=2$ that $n \geq 3k-2$.\\
Then, $\mathcal{C}$ is an MDS convolutional code.
\end{theorem}

\begin{proof}
First of all, note that the Singleton bound when $k \mid \delta$ is $(n-k)(\frac{\delta}{k}+1)+\delta+1 = \frac{n\delta}{k}+n-k+1.$\\
Let $u(z)\in\mathbb F_q[z]$ be a message with $\deg(u)=\ell$ and let $v(z)=u(z)G(z)$ be the corresponding codeword, i.e. $\deg(v)=\mu+\ell$ due to the predictable degree property, see \cite{bookchapter}. Then, one obtains
$(v_0\ v_1\ \cdots\ v_{\mu+\ell})=(u_0\ u_1\ \cdots\ u_{\ell})\mathcal{G},$
where 
\begin{align}
 \mathcal{G}&=\begin{pmatrix}
    G_0 & \cdots & G_{\mu}  & 0 & \cdots & 0\\
  0  &  \ddots & &  \ddots & \ddots & \vdots \\
   \vdots & \ddots &  \ddots & &  \ddots & 0\\
   0  & \cdots & 0 & G_0 & \cdots & G_{\mu}
\end{pmatrix}   \quad &\text{for}\quad \ell\geq \mu \nonumber\\ \vspace{10mm}\nonumber\\
 \mathcal{G}&=\begin{pmatrix}
    G_0 & \cdots & G_{\ell} & \cdots & G_{\mu} & & 0 \\
    & \ddots & \vdots &   & \vdots & \ddots \\
   0  &  &   G_0 & \cdots & G_{\mu-\ell} & \cdots & G_{\mu}
\end{pmatrix}   \quad &\text{for}\quad \ell<\mu \nonumber
\end{align}

Note that we can assume $u_0\neq 0$ for the calculation of the free distance because $wt(z^du(z)G(z))=wt(u(z)G(z))$ for any $d\in\mathbb N$.

For estimating the free distance we will split $\mathcal{G}$ into several matrices in different ways. If \bm{$\mu \geq 3$} we can distinguish three cases:

\textbf{Case 1: $\ell\geq\mu-1, \ell> 0$}
\begin{align*}
    wt(v(z))&\geq wt\left((u_0\ \cdots\ u_{\mu-1})\begin{pmatrix} G_0 & \cdots & G_{\mu-1}\\ & \ddots & \vdots\\ 0 & & G_0\end{pmatrix}\right)\\&+  wt\left((u_{\ell-\mu+1}\cdots u_{\ell})\begin{pmatrix}G_{\mu} & & 0\\ \vdots & \ddots &\\ G_1 & \cdots & G_{\mu} \end{pmatrix}\right)
\end{align*}
Through row and column permutations the matrix
$\begin{pmatrix}G_{\mu} & & 0\\ \vdots & \ddots &\\ G_1 & \cdots & G_{\mu} \end{pmatrix}$ can be transformed into $\bar{G}_{\mu-1}^c$ and therefore all full-size minors of both matrices are nonzero.
Note that $L \geq \mu$ for $k\mid \delta$. \\
Because all non trivially zero full-size minors of $G_{\mu-1}^c=\begin{pmatrix}
G_0 & \cdots & G_{\mu-1}\\
 & \ddots & \vdots\\
0 &  & G_0
\end{pmatrix}$ and 
$\bar{G}_{\mu-1}^c=\begin{pmatrix}
G_{\mu} & \cdots & G_{1}\\
 & \ddots & \vdots\\
0 &  & G_{\mu}\end{pmatrix}$ are nonzero, which implies that $G_0$ and $G_{\mu}$ are full rank, and we can apply Theorem \ref{mdpcrit} to $\mathcal{C}$ and $\overline{\mathcal{C}}$, respectively.

As $u_0\neq 0\neq u_{\ell}$, we can conclude:

$$
wt(v(z))\geq 2\big((n-k)\mu+1\big)= (n-k)\frac{\delta}{k}+1+(n-k)\frac{\delta}{k}+1.
$$

Therefore $wt(v(z))\geq (n-k)\Big(\frac{\delta}{k}+1\Big)+\delta+1= (n-k)\frac{\delta}{k}+(n-k)+\delta+1$ if and only if $(n-k)\frac{\delta}{k}+1 \geq (n-k)+\delta \Leftrightarrow (n-k)\Big(\frac{\delta-k}{k}\Big) \geq \delta -1 
\Leftrightarrow n \geq \frac{(\delta-1)k}{(\frac{\delta}{k}-1)k}+k=\frac{\delta-1}{\mu -1}+k$.





\textbf{Case 2: $1\leq \ell<\mu-1$}

In this case, we have
\begin{align*}
    &wt(v(z))\geq wt\left((u_0\ \cdots\ u_{\ell-1})\begin{pmatrix} G_0 & \cdots & G_{\ell-1}\\ & \ddots & \vdots\\ 0 & & G_0\end{pmatrix}\right)\\
&+wt\left((u_0\ \cdots\ u_{\ell})\begin{pmatrix}
G_{\ell} & \cdots & G_{\mu}\\
\vdots  &   & \vdots\\
G_0 & \cdots & G_{\mu-\ell}
\end{pmatrix}\right)+wt\left((u_{1}\ \cdots\ u_{\ell})\begin{pmatrix}G_{\mu} & & 0\\ \vdots & \ddots &\\ G_{\mu-\ell+1} & \cdots & G_{\mu} \end{pmatrix}\right).
\nonumber\end{align*}
Applying Lemma \ref{le2}, we obtain 
$$wt\left((u_0\ \cdots\ u_{\ell})\begin{pmatrix}
G_{\ell} & \cdots & G_{\mu}\\
\vdots  &   & \vdots\\
G_0 & \cdots & G_{\mu-\ell}
\end{pmatrix}\right)\geq n(\mu-\ell+1)-k(\ell+1)+1.$$
Using Lemma \ref{maxim}, one gets as in the previous case that all nontrivial full-size minors of $G_{\ell-1}^c$ and 
$\bar{G}_{\ell-1}^c$ are nonzero. Finally,
\begin{align*}
wt(v(z))&\geq 2\big((n-k)\ell+1\big)+n(\mu-\ell+1)-k(\ell+1)+1\\ &= 2\big((n-k)\ell+1\big) + \frac{n\delta}{k}-n\ell+n-k\ell-k+1.
\end{align*}
Therefore, $wt(v(z))\geq \frac{n\delta}{k}+n-k+1$ if $2\big((n-k)\ell+1\big) -n\ell-k\ell \geq 0  \Leftrightarrow 2n\ell-2k\ell+2-n\ell-k\ell \geq 0 \Leftrightarrow n \ell \geq 2k\ell -2 +k\ell\Leftrightarrow n \geq 3 k - \frac{2}{\ell}$.


If $n(\mu-\ell+1)-k(\ell+1)+1\leq 0$, the middle part of the weight and thus the condition on the matrix depending on $\ell$ is not needed. This explains the upper bound on $\ell$ in the assumptions of the theorem.

\textbf{Case 3: $\ell=0$}

In this case, $wt(v(z))=wt(u_0 (G_0\cdots G_{\mu}))\geq n(\mu+1)-k+1 = \frac{n \delta}{k}+n-k+1$.\\





In sum, if $\mu \geq 3$ we need $n \geq \frac{\delta-1}{\mu -1}+k$ for Case 1 and $n \geq 3k- \frac{2}{\ell}$ for Case 2. In order to be able to compare the two bounds obtained, note that for Case 2 when $\ell = \mu -2$ we obtain the largest bound. So we can rewrite the second bound as being $3k-\frac{2}{\mu-2}=3k-\frac{2}{\frac{\delta}{k}-2} = 3k - \frac{2k}{\delta-2k}.$

Consequently, comparing $\frac{(\delta-1)k}{\delta-k}+k$ and $3k- \frac{2k}{\delta - 2k}$ we observe that the second bound is always bigger than the first except for case $k=1$ and $\delta = 3$ and for the case $k=2$ and $\delta = 6$. Since we have to satisfy both cases simultaneously, we have $3k- \frac{2k}{\delta - 2k}$, except for the case $k=1$ and $\delta=3$ where $n \geq 2$ and for the case $k=2$ and $\delta = 6$ where $n \geq 9/2$.

For $\mu\leq 2$, we obtain a higher weight, i.e. smaller lower bound on $n$, if we put $\ell=\mu-1$ to Case 2 instead of Case 1. Also for $\mu\in\{0,1\}$ not all of the above casee exist.

If \bm{$\mu = 0$}, then $ \mathcal{G} = G_0$ and therefore the only relevant case is when $\ell = 0$. Note that if $\mu = 0$ then $\delta = 0$ and then the Singleton bound can be rewritten as $n-k+1$. Additionally, we easily obtain that $wt(v(z))=wt(u_0G_0)\geq n-k+1$ without any restriction on the parameters.


If \bm{$\mu = 1$}, then $\delta = k$ and the Singleton bound can be rewritten as $2n-k+1$. We can distinguish two cases:\\
\textbf{Case 1:} $\ell \geq \mu$\\
$wt(v(z))\geq wt\left(\left((u_0\ u_{1})\begin{pmatrix}
G_{0} & G_{1}\\
0 & G_{0}
\end{pmatrix}\right)\right) + wt(u_1G_1) \geq 2(n-k)+1+n-k+1 = 3n-3k+2\geq 2n-k+1$ if $n \geq 2k - 1.$\\
\textbf{Case 2:} $\ell=0$\\
In this case $ \mathcal{G} = (G_0 \ G_1)$ and we easily obtain that $wt(v(z))=wt(u_0(G_0 \ G_1))= 2n-k+1.$\\
In sum, if $\mu = 1$ we need $n \geq 2k-1$.

If \bm{$\mu = 2$}, then $\delta = 2k$ and the Singleton bound can be rewritten as $3n-k+1$. We can distinguish three cases:\\
\textbf{Case 1:} $\ell \geq \mu$
\begin{align*}
    wt(v(z))&\geq wt\left((u_0\ u_{1}\ u_2)\begin{pmatrix}
G_{0} & G_{1} & G_2\\
0 & G_{0} & G_1\\
0 & 0 & G_0
\end{pmatrix}\right) + wt\left((u_{\ell-1} \ u_{\ell})\begin{pmatrix}
G_{2} & 0\\
G_{1} & G_2
\end{pmatrix}\right)\\
&\geq 3 (n-k)+1+2(n-k)+1 = 5n-5k+2 \geq 3n-k + 1
\end{align*}
 if $n \geq 2k-\frac{1}{2}$.

\textbf{Case 2:} $0 < \ell \leq \mu -1$\\
In this case $\mu = 2$ implies $\ell = 1$ and $ \mathcal{G}=\begin{pmatrix}
G_{0} & G_{1} & G_2 & 0\\
0 & G_{0} & G_{1} & G_2
\end{pmatrix}$. Therefore, $wt(v(z))=wt(u_0G_0) + wt\left((u_0\ u_{1})\begin{pmatrix}
G_{1} & G_{2}\\
G_{0} & G_{1}
\end{pmatrix}\right) + wt(u_1G_2) \geq (n-k+1)+(2n-2k+1)+(n-k+1) = 4n-4k+3 \geq 3n-k + 1$ if $n \geq 3k-2$.

\textbf{Case 3:} $\ell=0$\\
In this case $ \mathcal{G} = (G_0 \ G_1 \ G_2)$ and we easily obtain that\\ $wt(v(z))=wt(u_0(G_0 \ G_1 \ G_2))\geq 3n-k+1.$

In sum, if $\mu = 2$, we need $n \geq  2k-\frac{1}{2}$ and $n \geq 3k-2$. Therefore it is enough to have $n \geq 3k-2$, except for $k = 1$ where $n \geq2 $.
\end{proof}

\begin{remark}
We have $L \geq \mu$ if $k\mid \delta$, i.e. codes fulfilling the conditions of the above theorem are not only MDS but also reach the upper bound for the $j$-th column distance and the $j$-th reverse column distance until $j=\mu-1$.
\end{remark}

\subsection{Criteria for MDS convolutional codes with $k\nmid\delta$}

\begin{theorem}\label{main2}
Let $n,k,\delta\in\mathbb N$ with $n>k\nmid\delta$ and let $\mathcal{C}$ be an $(n,k,\delta)$ convolutional code over $\mathbb F_q$ with minimal generator matrix $G(z)$ of degree $\mu=\lceil\frac{\delta}{k}\rceil$ and with generic row degrees. Denote by $\tilde{G}_{\mu}$ the matrix consisting of the (first) $t=\delta+k-k\mu$ nonzero rows of $G_{\mu}$. 
Moreover, set
\begin{align*}
B_1&=t-1+\frac{\mu-1}{2}k+\frac{\delta}{\mu}\ \text{for $\mu\geq 2$ and $B_1=0$ otherwise}\\
B_2&=-1+\frac{\mu}{2}k+\frac{\delta}{\mu-1}\ \text{for $\mu\geq 2$ and $B_2=0$ otherwise}\\
B_3&=\left(2-\frac{\mu}{2}\right)k+\delta-1-\frac{\mu k+k-\delta+1}{\mu-1}\\
B_4&=-1+\left(2+\frac{\mu-1}{2}\right)k-\frac{1}{\mu-2}\\
B_5&=t-1+(1-\mu)k+\delta=2(\delta+k(1-\mu))-1\\
B&=\max(B_1,B_2,B_3,B_4,B_5)
\end{align*}

If all not trivially zero full-size minors of the matrices
\begin{align*}
 &\begin{pmatrix} G_0 & \cdots & G_{\mu-1}\\ & \ddots & \vdots\\ 0 & & G_0\end{pmatrix}\  \text{and}\  \begin{pmatrix}
G_{\ell} & \cdots & G_{\mu-1}\\
\vdots  &   & \vdots\\
G_0 & \cdots & G_{\mu-1-\ell}
\end{pmatrix} \ \text{for}\ \ 0\leq \ell<\mu-1\\
&\text{and}\ \begin{pmatrix}
    \tilde{G}_{\mu}\\ G_{\mu-1}\\  \vdots\\ G_i
\end{pmatrix}  \   \text{for}\ \ 0< i\leq\mu-1\ \text{s.t.}\ n\geq k(\mu-i+1)
\  \text{and}\ \tilde{G}_{\mu}
\end{align*}
are nonzero and $n\geq B$, then $\mathcal{C}$ is MDS.
\end{theorem}

\begin{proof}
We proceed similar to the proof of Theorem \ref{main} but need to distinguish more cases as some rows of $G_{\mu}$ are equal to zero. Let $u(z)\in\mathbb F_q[z]$ be a message with $\deg(u)=\ell$ and let $v(z)=u(z)G(z)$ and we can assume $u_0\neq 0\neq u_{\ell}$. Using again Lemma \ref{le2} and Theorem \ref{mdpcrit}, we obtain the following.

    \textbf{Case 1: $\ell\geq\mu-1> 0$}
\begin{align*}
    wt(v(z))&\geq wt\left((u_0\ \cdots\ u_{\mu-1})\begin{pmatrix} G_0 & \cdots & G_{\mu-1}\\ & \ddots & \vdots\\ 0 & & G_0\end{pmatrix}\right)+ \\
    &+\sum_{i=1}^{\mu-1} wt\left((u^{(1)}_{\ell-\mu+i}\ u_{\ell-\mu+i+1}\ \cdots\ u_{\ell})\begin{pmatrix}
    \tilde{G}_{\mu}\\ \vdots\\ G_i
\end{pmatrix}\right)+wt(u^{(1)}_{\ell}\tilde{G}_{\mu})
\end{align*}
 
where for $j=0,\hdots,\ell$, $u_j=[u_j^{(1)}\ u_j^{(2)}]$ with $u_j^{(1)}\in\mathbb F_q^{t}$.

Case 1.1 : $u_{\ell}^{(1)}\neq 0$.
\begin{align*}
    wt(v(z))&\geq \mu(n-k)+1+\sum_{i=1}^{\mu-1}(n-(\mu-i)k-t+1)+n-t+1\\
    &=\mu (n-k)+\delta+1+\mu n-\frac{(\mu-1)\mu}{2}k-\mu t+\mu-\delta\\
    &\geq \mu (n-k)+\delta+1
\end{align*}
if
$$n\geq t-1+\frac{\mu-1}{2}k+\frac{\delta}{\mu}=:B_1.$$

Case 1.2: $u_{\ell}^{(1)}=0$.
\begin{align*}
    wt(v(z))&\geq \mu(n-k)+1+\sum_{i=1}^{\mu-1}(n-(\mu-i)k+1)\\
    &=\mu (n-k)+\delta+1+(\mu-1)n-\frac{(\mu-1)\mu}{2}k+\mu-1-\delta\\
    &\geq \mu (n-k)+\delta+1
\end{align*}
if
$$n\geq -1+\frac{\mu}{2}k+\frac{\delta}{\mu-1}=:B_2.$$


 \textbf{Case 2: $0<\ell<\mu-1$}

\begin{align*}
    wt(v(z))&\geq wt\left((u_0\ \cdots\ u_{\ell-1})\begin{pmatrix} G_0 & \cdots & G_{\ell-1}\\ & \ddots & \vdots\\ 0 & & G_0\end{pmatrix}\right)\\
&+wt\left((u_0\ \cdots\ u_{\ell})\begin{pmatrix}
G_{\ell} & \cdots & G_{\mu-1}\\
\vdots  &   & \vdots\\
G_0 & \cdots & G_{\mu-\ell-1}
\end{pmatrix}\right)\nonumber\\
&+\sum_{i=\mu-\ell}^{\mu-1} wt\left((u^{(1)}_{\ell-\mu+i}\cdots u_{\ell})\begin{pmatrix}
    \tilde{G}_{\mu}\\ \vdots\\ G_i
\end{pmatrix}\right)+wt(u^{(1)}_{\ell}\tilde{G}_{\mu})
\nonumber\end{align*}

Case 2.1 : $u_{\ell}^{(1)}\neq 0$
\begin{align*}
    &wt(v(z))\geq \ell(n-k)+1+n(\mu-\ell)-k(\ell+1)+1+\sum_{i=1}^{\ell+1}(n-(i-1)k-t+1)\\
    &=\mu (n-k)+\delta+1+(\ell+1)n-\left(2\ell+1+\frac{(\ell+1)\ell}{2}-\mu\right)k+\ell+2-t(\ell+1)-\delta\\
    &\geq \mu (n-k)+\delta+1
\end{align*}
if
$$n\geq t-1+\left(2-\frac{\mu+1}{\ell+1}+\frac{\ell}{2}\right)k+\frac{\delta-1}{\ell+1}.$$
Note that this lower bound on $n$ is strictly monotonic increasing with $\ell$ and thus,
\begin{align*}
&B_3:=\max_{\ell\in\{1,...,\mu-2\}}t-1+\left(2-\frac{\mu+1}{\ell+1}+\frac{\ell}{2}\right)k+\frac{\delta-1}{\ell+1}\\
&=t-1+\left(2-\frac{\mu+1}{\mu-1}+\frac{\mu-2}{2}\right)k+\frac{\delta-1}{\mu-1}
=\left(2-\frac{\mu}{2}\right)k+\delta-1-\frac{\mu k+k-\delta+1}{\mu-1}
\end{align*}

Case 2.2: $u_{\ell}^{(1)}=0$
\begin{align*}
    wt(v(z))&\geq \ell(n-k)+1+n(\mu-\ell)-k(\ell+1)+t+1+\sum_{i=2}^{\ell+1}(n-(i-1)k+1)\\
    &=\mu (n-k)+\delta+1+\ell n-\left(2\ell+1+\frac{(\ell+1)\ell}{2}-\mu\right)k+\ell+1+t-\delta\\
    &\geq \mu (n-k)+\delta+1
\end{align*}
if
$$n\geq -1+\left(2-\frac{\mu-1}{\ell}+\frac{\ell+1}{2}\right)k+\frac{\delta-1-t}{\ell}.$$
Also this lower bound on $n$ is strictly monotonic increasing with $\ell$ and thus,
\begin{align*}
    B_4&:=\max_{\ell\in\{1,...,\mu-2\}}-1+\left(2-\frac{\mu-1}{\ell}+\frac{\ell+1}{2}\right)k+\frac{\delta-1-t}{\ell}\\
    &=-1+\left(2+\frac{\mu-1}{2}\right)k-\frac{1}{\mu-2}
\end{align*}

 \textbf{Case 3: $\ell=0$}
\begin{align}
    wt(v(z))&= wt(u_0\begin{pmatrix} G_0 & \cdots & G_{\mu-1}\end{pmatrix})
+wt(u^{(1)}_{0}\tilde{G}_{\mu})
\nonumber\end{align}

Case 3.1 : $u_{\ell}^{(1)}\neq 0$
\begin{align*}
    wt(v(z))&\geq n\mu-k+1+n-t+1\\
    &=\mu (n-k)+\delta+1+n-(1-\mu)k-t+1-\delta\\
    &\geq \mu (n-k)+\delta+1
\end{align*}
if
$$n\geq t-1+(1-\mu)k+\delta=2(\delta+k(1-\mu))-1=:B_5.$$

Case 3.2: $u_{\ell}^{(1)}=0$
\begin{align*}
    wt(v(z))&\geq n\mu-k+t+1=\mu (n-k)+\delta+1
\end{align*}

Note that $k\nmid\delta$ implies $\delta>0$ and hence $\mu>0$. If $\mu\leq 2$, Case 2 does not exist and if $\mu=1$, additionally, Case 1 does not exist because Case 3 gives the condition that $G_0$ and $\tilde{G}_1$ have all full-size minors nonzero and thus, this needs to give enough weight also for $\ell>0$ (that $G_0$ and $\tilde{G}_{\mu}$ have all full-size minors nonzero is part of the conditions for all $\ell$).
\end{proof}

\begin{remark}
We have $L \geq \mu-1$ if $k\nmid \delta$, i.e. codes fulfilling the conditions of the above theorem reach the upper bound for the $j$-th column distance until $j=\mu-1$. Moreover, $L =\mu-1$ if and only if $n>\delta+k=k\mu+t$. In this case $\mathcal{C}$ is also MDP (and all matrices in the above theorem have more columns than rows).
\end{remark}



\section{Releasing conditions as far as possible to get a field size as small as possible}

The mentioned theorems of \cite{mdssr} use similar ideas to the ones of the preceding section but always split $\mathcal{G}$ fully into blocks of columns consisting of $n$ columns each (except for $\ell=0$ in which case this is not possible). This has the effect that these codes do not have good column distances as for each $\ell$, $G^c_{\ell}$ is split into blocks consisting of $n$ columns each. Moreover, each additional splitting of $\mathcal{G}$ increases the bound on $n$, i.e. the necessary rate. However, if $n,k,\delta$ are given and $n$ is large enough to allow to split $\mathcal{G}$ more than in Theorem \ref{main} or Theorem \ref{main2}, respectively, further splittings release the conditions and therefore allow constructions over fields of smaller size.\\
Given $n,k,\delta$ fulfilling the condition of Theorem \ref{main} or Theorem \ref{main2}, we want to find the optimal way to split the matrix $\mathcal{G}$. Each splitting releases the condition to get an MDS code, i.e. decreases the field size but it also decreases the weight of some codewords and we have to ensure that the code stays MDS. Moreover, if further splitting of $\mathcal{G}$ is possible, there might be different possible ways to split it. We want to do the splitting in such a way that we give priority to splittings with good column distances. If $k\mid\delta$, we can also achieve good reverse column distances with our method.\\
Denote by $S$ the value of the generalized Singleton bound. 

\subsection{Conditions for $k\mid\delta$}

First calculate 
$$W:=\left\lceil\frac{S-2}{n-k}\right\rceil=\left\lfloor\frac{\delta}{k}\right\rfloor+1+\left\lceil\frac{\delta-1}{n-k}\right\rceil\quad\text{and}\quad E:=\left\lceil\frac{W}{2}\right\rceil-1\quad\text{and}\quad F:=\left\lfloor\frac{W}{2}\right\rfloor-1.$$
Note that we assumed that $n,k,\delta$ fulfill the conditions of Theorem \ref{main} and hence, $G^c_{\mu}$ and $\bar{G}^c_{\mu-1}$ give rise to enough weight to reach the Singleton bound. Thus, $E+F+2=W\leq 2\mu+1$ and the largest values $E$ and $F$ can have are $E=\mu$ and $F=\mu-1$.
If all non-trivially zero full-size minors of $G_{E}^c$ and $\bar{G}_{F}^c$ are nonzero, then $wt(u(z)G(z))\geq S+R$ for all $u(z)\in\mathbb F[z]^k$ with $\ell=\deg(u)\geq E$, where $R=W(n-k)-(S-2)\in\{0,\hdots,n-k+1\}$. If $R\geq F\cdot k-1$, we can split $\bar{G}_{F}^c$ into $\bar{G}_{F-1}^c$ and $\begin{pmatrix}
    G_{\mu-F}\\ \vdots \\ G_{\mu}
\end{pmatrix}$ and if $R-F\cdot k+1\geq E\cdot k-1$, we can additionally split $G_{E}^c$ into $G_{E-1}^c$ and $\begin{pmatrix}
    G_E\\ \vdots\\ G_0
\end{pmatrix}$.
If such a splitting is possible, by the choice of $W,E,F$, the part of the weight arising from $\begin{pmatrix}
  G_{\mu}\\ \vdots \\ G_{\mu-F}
\end{pmatrix}$, i.e. $wt\left((u_{\ell-F}\ \cdots\ u_{\ell})\begin{pmatrix}
  G_{\mu}\\ \vdots \\ G_{\mu-F}
\end{pmatrix}\right)$, is necessary, which implies $n\geq (F+1)k$.

It remains to consider the case that $\ell<E$.

If $\ell<F\leq\mu-1$, we write, using the splitting of Case 2 in Theorem \ref{main},
$$wt(v(z))=S+A$$
 with $A:=n\ell-3k\ell+2$
and consider the following cases:




If $A\geq k$, we can change the splitting of $\mathcal{G}$ to $G_{\ell}^c$, $\begin{pmatrix}
    G_{\ell+1} & \cdots & G_{\mu}\\
    \vdots & & \vdots\\
    G_1 & \cdots & G_{\mu-\ell}
\end{pmatrix}$, $\bar{G}_{\ell-1}^c$. If even $A\geq 2k$, we can consider the splitting $G_{\ell}^c$, $\begin{pmatrix}
    G_{\ell+1} & \cdots & G_{\mu-1}\\
    \vdots & & \vdots\\
    G_1 & \cdots & G_{\mu-\ell-1}
\end{pmatrix}$, $\bar{G}_{\ell}^c$. Finally, if we still have some weight left, we can split the middle matrix. With each splitting, we loose weight $(\ell+1)k-1$ if all occurring matrices have at least as many columns as rows, otherwise we can leave away the matrices with fewer columns than rows and might loose less than weight $(\ell+1)k-1$ per splitting.\\
Hence, one can split the middle matrix $x$ times, where 
$$x=\min\left(\mu-\ell-2, \left\lfloor\frac{A-2k}{(\ell+1)k-1}\right\rfloor\right).$$
If $x=\mu-\ell-2$, i.e. we are able to split the middle part fully into blocks consisting of $n$ columns each, we can delete up to $y$ matrices (with $n$ columns each) and do not need to consider their full-size minors, where 
$$y=\min\left(\mu-\ell-1, \left\lfloor\frac{A-2k-(\mu-\ell-2)((\ell+1)k-1)}{n-(\ell+1)k+1}\right\rfloor\right).$$
If $\ell=F=E-1$ and $\mu\geq 3$, then $\ell=\mu-1$ belongs to Case 1 in the proof of Theorem \ref{main} and hence we know $E\leq\mu-1$, i.e. $\ell\leq\mu-2$, and can proceed in the following way: If $\bar{G}_{F}^c$ has not been split in a previous step, we can proceed exactly as before. If it has been split before, we first see if the splitting $G_{\ell}^c$, $\begin{pmatrix}
    G_{\ell+1} & \cdots & G_{\mu-1}\\
    \vdots & & \vdots\\
    G_1 & \cdots & G_{\mu-\ell-1}
\end{pmatrix}$, $\bar{G}_{\ell}^c$ is possible, then if we can additionally split $\bar{G}_{\ell}^c=\bar{G}_{F}^c$ into $\bar{G}_{F-1}^c$ and $\begin{pmatrix}
    G_{\mu-F}\\ \vdots \\ G_{\mu}
\end{pmatrix}$ and finally, we investigate, which splitting of the middle part is possible.

If $\ell=F=E-1$ and $\mu\in\{0,1\}$, the case $0<\ell\leq\mu-1$ does not exist and hence, one only has Case 1 and Case 3. For Case 3, one has to determine $E$ and $F$, see if we can split $\bar{G}_{F}^c$ into $\bar{G}_{F-1}^c$ and $\begin{pmatrix}
    G_{\mu-F}\\ \vdots \\ G_{\mu}
\end{pmatrix}$ and if we can split $G_{E}^c$ into $G_{E-1}^c$ and $\begin{pmatrix}
    G_E\\ \vdots\\ G_0
\end{pmatrix}$. For Case 3, there is nothing to do as (for $k\mid\delta$) the matrix $[G_0\ \cdots\ G_{\mu}]$ gives exactly the weight of the Singleton bound and can never be split.

If $\ell=F=E-1$ and $\mu=2$, we have to additionally consider the case $\ell=1=F=E-1$, in which we start with the splitting $\begin{pmatrix} G_0 & \vline & G_1 & G_2 & \vline & \\
& \vline & G_0 & G_1 & \vline & G_2
\end{pmatrix}$. And with the given parameters it is clear that we have to stay with this splitting as $\begin{pmatrix} G_0 &  G_1 & \vline &  G_2 & \\
&  G_0 & \vline &G_1 &  G_2
\end{pmatrix}$ would imply $E=1$, but in the considered case one has $E=2$.

We illustrate this procedure with an example.

\begin{example}
    Let $k=2$, $n=11$, $\delta=6$, i.e. $\mu=3$ and $S=43$. We calculate $W=5$, i.e. $E=2$, $F=1$. Then, $R=4\geq Fk-1+Ek-1$, i.e. from the case $\ell\geq E$ we obtain the matrices 
    $$\begin{pmatrix}
        G_0 & G_1\\ 0 & G_0
    \end{pmatrix}, \ \begin{pmatrix}
        G_2\\ G_1\\ G_0
    \end{pmatrix},\ \begin{pmatrix}
        G_2\\ G_3
    \end{pmatrix},\ G_3.$$
    For $\ell=1=F=E-1$, we start with the splitting $G_0$, $\begin{pmatrix}
        G_1 & G_2 & G_3\\ G_0 & G_1 & G_2
    \end{pmatrix}$, $G_3$. Since $A=7\geq 2k$, we can change to the splitting  $\begin{pmatrix}
        G_0 & G_1\\ 0 & G_0
    \end{pmatrix}, \ \begin{pmatrix}
        G_2\\ G_1
    \end{pmatrix},\ \begin{pmatrix}
      G_3 & 0\\  G_2 &  G_3
    \end{pmatrix}$ and since $A-2k=4\geq Fk-1$, we can obtain the splitting
    $\begin{pmatrix}
        G_0 & G_1\\ 0 & G_0
    \end{pmatrix}, \ \begin{pmatrix}
        G_2\\ G_1
    \end{pmatrix},\  \begin{pmatrix}
        G_2\\ G_3
    \end{pmatrix},\ G_3$.
    As $\mu-\ell-2=0$, we get $x=0$, i.e. we cannot split the middle part $\begin{pmatrix}
        G_2\\ G_1
    \end{pmatrix}$ further, which is obvious since it already consists only of one block. As $A-2k-(Fk-1)=3<8=n-(\ell+1)k+1$, also $y=0$ and the middle part $\begin{pmatrix}
        G_2\\ G_1
    \end{pmatrix}$ cannot be left away. Also including the case that $\ell=0$, one obtains in sum that the non trivially zero full-size minors of the following matrices have to be nonzero:
     $$\begin{pmatrix}
        G_0 & G_1\\ 0 & G_0
    \end{pmatrix}, \ \begin{pmatrix}
        G_2\\ G_1\\ G_0
    \end{pmatrix},\ \begin{pmatrix}
        G_2\\ G_3
    \end{pmatrix}, \begin{pmatrix}
        G_2\\ G_1
    \end{pmatrix},\ [G_0\ G_1\ G_2\ G_3].$$
    Note that we can omit $G_3$ as its full-size minors are part of the full-size minors of $[G_0\ G_1\ G_2\ G_3]$.
\end{example}

\subsection{Conditions for $k\nmid\delta$}

First calculate $D$ maximal such that $n\geq Dk+t$ and then $E$ minimal such that 
\begin{align*}
    (n-k)(E+1)+1+\sum_{i=1}^{\min(E+1,D+1)}(n-(i-1)k-t+1)\geq S
\end{align*}
If $D\geq E$, choose $F$ minimal such that 
\begin{align*}
    (n-k)(E+1)+1+\sum_{i=1}^{F+1}(n-(i-1)k-t+1)\geq S
\end{align*}

If $D<E$, set $F:=D$.

If $\ell\geq E$, we have the splitting $G_E^c$, $\begin{pmatrix}
    \tilde{G}_{\mu}\\  G_{\mu-1}\\  \vdots \\ G_{\mu-F}
\end{pmatrix},\hdots,   \tilde{G}_{\mu}$ and additionally check if it is possible to split $G_{E}^c$ into $G_{E-1}^c$ and $\begin{pmatrix}
    G_E\\ \vdots\\ G_0
\end{pmatrix}$ and still have enough weight. Note that if such a splitting is possible, this implies $n\geq (E+1)k$ by the definition of $E$. This then already determines the optimal splitting.

Let 
\begin{align*}
    A_1&=(\ell+1)n-\left(2\ell+1+\frac{(\ell+1)\ell}{2}-\mu\right)k+\ell+2-t(\ell+1)-\delta\\
    A_2&= \ell n-\left(2\ell+1+\frac{(\ell+1)\ell}{2}-\mu\right)k+\ell+1+t-\delta
\end{align*}
Note that $A_1\geq A_2$ if and only if $n\geq t(\ell+2)-1$.

If $\ell<E$,
we start with the splitting $G_{\ell-1}^c$, $\begin{pmatrix}
G_{\ell} & \cdots & G_{\mu-1}\\
\vdots  &   & \vdots\\
G_0 & \cdots & G_{\mu-\ell-1}
\end{pmatrix}$,\\  $\begin{pmatrix}
    \tilde{G}_{\mu}\\ \vdots \\ G_{\mu-\min(D,\ell)}
\end{pmatrix},\hdots,   \tilde{G}_{\mu}$. If $A_1\geq k$ and $A_2\geq t$, we can change it to the splitting $G_{\ell}^c$, $\begin{pmatrix}
G_{\ell+1} & \cdots & G_{\mu-1}\\
\vdots  &   & \vdots\\
G_1 & \cdots & G_{\mu-\ell-1}
\end{pmatrix}$,  $\begin{pmatrix}
    \tilde{G}_{\mu}\\ \vdots \\ G_{\mu-\min(D,\ell)}
\end{pmatrix},\hdots,  \tilde{G}_{\mu}$. Afterwards, we can split the matrix $\begin{pmatrix}
G_{\ell+1} & \cdots & G_{\mu-1}\\
\vdots  &   & \vdots\\
G_1 & \cdots & G_{\mu-\ell-1}
\end{pmatrix}$ into 
$$x=\min\left(\mu-\ell-2, \left\lfloor\frac{A_1-k}{(\ell+1)k-1}\right\rfloor, \left\lfloor\frac{A_2-t}{\ell k+t-1}\right\rfloor\right)$$
parts.

If $x=\mu-\ell-2$, i.e. we are able to split the middle part fully into blocks consisting of $n$ columns each, we can delete up to $y$ matrices (with $n$ columns each) and do not need to consider their full-size minors, where
\begin{align*}
y&= \min\left(\mu-\ell-1, C_1, C_2\right)\\
    C_1&=\left\lfloor\frac{A_1-k-(\mu-\ell-2)((\ell+1)k-1)}{n-(\ell+1)k+1}\right\rfloor\\
    C_2&=\left\lfloor\frac{A_2-t-(\mu-\ell-2)(\ell k+t-1)}{n-\ell k-t+1}\right\rfloor
\end{align*}
If $F<\min(D,\ell)$, one can additionally check how many of the matrices $$\begin{pmatrix}
    \tilde{G}_{\mu}\\ \vdots \\ G_{\mu-\min(D,\ell)}
\end{pmatrix}, \hdots, \begin{pmatrix}
    \tilde{G}_{\mu}\\ \vdots \\ G_{\mu-F-1}
\end{pmatrix}$$ can be removed.

\begin{remark}
 In all cases, we only considered splittings where we do not split inside the coefficient matrices, i.e. we do only split into blocks consisting of at least $n$ columns. In principle, one might also be able to split more if $n$ is very large but covering also this would make the above description even more technically complicated and applies anyway only to very large (i.e. unpractical) rates.
\end{remark}


\section{Construction of MDS convolutional codes}

In this section, we apply the criteria of the preceding sections for the construction of MDS convolutional codes. In the first subsection, we present general constructions fulfilling the criteria of Theorem \ref{main} and Theorem \ref{main2}, in the second subsection, we give some concrete construction examples over small finite fields.

\subsection{General constructions}

If $k\mid\delta$, to find constructions for generator matrices that fulfill the properties of Theorem \ref{main}, we can make use of so-called reverse superregular matrices as defined in the following.

\begin{definition}
Let $r,n,m\in\mathbb N$ and consider a Toeplitz matrix $A\in\mathbb F_q^{(r+1)n\times(r+1)m}$ of the form $A=\begin{pmatrix}A_0 & \cdots & A_r\\ & \ddots & \vdots\\ 0 & & A_0\end{pmatrix}$ with $A_i\in\mathbb F_q^{n\times m}$ for $i\in\{0,\hdots,r\}$.  $A$ is called \textbf{(reverse) superregular Toeplitz matrix} if all non trivially zero minors (of any size) of the matices $A$ and $A_{rev}=\begin{pmatrix}A_r & \cdots & A_0\\ & \ddots & \vdots\\ 0 & & A_r\end{pmatrix}$ are nonzero.
\end{definition}

As all the matrices of the form $\begin{pmatrix}
G_{\ell} & \cdots & G_{\mu}\\
\vdots  &   & \vdots\\
G_0 & \cdots & G_{\mu-\ell}
\end{pmatrix}$ with $0\leq\ell\leq\mu$ are submatrices of $G_{\mu}^c$, all conditions of Theorem \ref{main} are fulfilled if $G_{\mu}^c$ is a reverse superregular Toeplitz matrix.
Hence, we can use existing constructions for reverse superregular Toeplitz matrices to construct MDS convolutional codes. Such constructions can e.g. be found in \cite{virtuphD} or \cite{li17}. 
These constructions can be modified to be also used for the case that $k\nmid\delta$. In the following theorem we write down such a modified version for one of the constructions of \cite{li17}.
\begin{theorem}\ \\
Let $n,k,\delta\in\mathbb N$ such that they either fulfill the conditions of Theorem \ref{main} or Theorem \ref{main2} and let $\alpha$ be a primitive element of a finite field $\mathbb F=\mathbb F_{p^N}$ with $N>\mu\cdot 2^{(\mu+1)n+t-1}$. Then $G(z)=\sum_{i=0}^{\mu}G_iz^i$ with $G_i=\left[\begin{array}{ccc}
\alpha^{2^{in}} & \hdots & \alpha^{2^{(i+1)n-1}} \\ 
\vdots &  & \vdots \\ 
\alpha^{2^{in+k-1}} & \hdots & \alpha^{2^{(i+1)n+k-2}}
\end{array}\right]$ for $i=0,\hdots,\mu-1$ and $\tilde{G}_{\mu}=\begin{pmatrix}\alpha^{2^{\mu n}} & \hdots & \alpha^{2^{(\mu+1)n-1}} \\ 
\vdots &  & \vdots \\ 
\alpha^{2^{\mu n+t-1}} & \hdots & \alpha^{2^{(\mu+1)n+t-2}}\end{pmatrix}$ is the generator matrix of an MDS convolutional code.
\end{theorem}

\begin{proof}
    The only thing that needs to be done in addition to the proof in \cite{li17} is the calculation of the bound for the field size. To this end, we need to estimate the value of the highest exponent of $\alpha$ occurring in any of the full-size minors considered in Theorem \ref{main2}. The highest exponent of $\alpha$ occurring in any of the full-size minors of $G_i$ is $2^{(i+1)n+k-2k}\sum_{j=0}^{k-1}(2^2)^j<2^{(i+1)n+k-1}$ for $i=0,\hdots,\mu-1$ and $2^{(\mu+1)n+t-2t}\sum_{j=0}^{t-1}(2^2)^j<2^{(\mu+1)n+t-1}$ in any of the full-size minors of $\tilde{G}_{\mu}$. Hence, the highest exponent of $\alpha$ occurring in any of the full-size minors considered in Theorem \ref{main2} is smaller than $\mu\cdot 2^{(\mu+1)n+t-1}$.
\end{proof}

However, the field sizes obtained from these constructions are very large. In the following section, we will present some examples for small code parameters over small finite fields using our new results directly.

\subsection{Construction examples for small code parameters}

In the following, we give some examples for small code parameters, where the criteria of Theorem \ref{main} and Theorem \ref{main2} together with the results of Section 4 allow the construction of MDS convolutional codes over smaller fields than required for already existing constructions (up to our knowledge). We start with some examples where $k\mid\delta$.

\begin{example}
If $k=\delta=1$, i.e. $\mu=1$, according to Theorem \ref{main}, $n$ can be arbitrary, and all non trivially zero full-size minors of $\begin{pmatrix}
    G_0 & G_1 \\ 0 & G_0
\end{pmatrix}$ and $(G_0\ G_1)$ have to be nonzero. 
If we want to release these conditions as far as possible according to the method described in Section 4.1, we first have to calculate $W=2$ (in this special case independent of $n$), i.e. $E=F=0$, telling us that it is indeed enough if all full-size minors, i.e. all entries, of $G_0$ and $G_1$ are nonzero. This means $G_0=G_1=(1\ \cdots\ 1)$ defines an MDS convolutional code over any field. 
\end{example}

\begin{example}
For $k=1$ and $\delta=2$, $n$ can be arbitrary according to Theorem \ref{main} and 
all non trivially zero full-size minors of $(G_0\ G_1\ G_2)$, $\begin{pmatrix}G_0 & G_1 & G_2\\ 0 & G_0 & G_1\\ 0 & 0 & G_0\end{pmatrix}$, $\begin{pmatrix}G_2 & G_1\\ 0 & G_2\end{pmatrix}$ and $\begin{pmatrix}G_1 & G_2\\ G_0 & G_1\end{pmatrix}$ have to be nonzero. According to Section 4.1 we calculate $W=4$ and $E=F=1$, i.e. it is indeed enough if all non trivially zero full-size minors of $(G_0\ G_1\ G_2)$, $\begin{pmatrix}G_0 & G_1\\ 0 & G_0\end{pmatrix}$ and $\begin{pmatrix}G_2 & G_1\\ 0 & G_2\end{pmatrix}$ are nonzero. 
Hence, an $(n,1,2)$ MDS convolutional code exists for 
$q\geq n+1$, e.g. $G_0=G_2=(1\ \cdots\ 1)$ and $G_1=(1\ \alpha\ \cdots \alpha^{n-1})$ where $\alpha$ is a primitive element of $\mathbb F_q$. For $n=2$ this field size is smaller than in existing constructions, for $n\geq 3$ it is equal to the best existing construction \cite{gll06}.
\end{example}


\begin{example}
For $k=1$, $n=\delta=3$, i.e. $\mu=3$  and $S=12$, the best existing constructions require $q\geq 10$ (see \cite{ju75,gll06}).
   Using Theorem \ref{main}, the $(3,1,3)$ convolutional code over 
   $\mathbb{F}_{16}$ with generator matrix $G(z)=\sum_{i=0}^{3}G_iz^i$, such that $G_0= (\alpha^2+1 \ \ 1 \ \ \alpha^3+1)$, $G_1= (\alpha^3+\alpha \ \ \alpha^3+\alpha^2+1 \ \ \alpha^3)$, $G_2= (\alpha^3+\alpha^2+\alpha+1 \ \ \alpha+1 \ \ \alpha^3+\alpha^2+\alpha)$ and $G_3= (\alpha^2+1 \ \  \alpha^3+\alpha^2 \ \ \alpha^2+\alpha+1)$, with $\alpha$ as primitive element of $\mathbb F_{16}$, is an MDS convolutional code.\\
To improve the field size, we proceed as described in Section 4.1 to release the conditions and calculate $W=5$, i.e. $E=2$ and $F=1$ giving rise to the matrices $G_2^c$ and $\overline{G}_1^c$, which lead to weight $7$ and $5$, respectively, i.e. cannot be split further.
For $\ell=1$, we obtain the matrices $G_1^c$ and $\overline{G}_1^c$ and $\begin{pmatrix}
    G_2\\ G_1
\end{pmatrix}$ and for $\ell=0$ the matrix $[G_0\ G_1\ G_2\ G_3]$. In summary, we need that the non trivially zero full-size minors of the following matrices are nonzero:
$G_2^c$, $\overline{G}_1^c$, $\begin{pmatrix}
    G_2\\ G_1
\end{pmatrix}$, $[G_0\ G_1\ G_2\ G_3]$.
Using this, we found an $(3,1,3)$ MDS convolutional code over $\mathbb F_7$ defined by the generator matrix $G(z)=\sum_{i=0}^{3}G_iz^i$, with $G_0 = ( 4 \ 4 \ 2),$ $G_1 = (1 \ 4 \ 3)$, $G_2 = (4\ 6 \ 2)$ and $G_3 = (1 \ 2 \ 1)$, which additionally has optimal $j$-th column distance for $j\leq 2$ and optimal reverse column distance for $j\leq 1$.
\end{example}

\begin{example}
Take $k = 2,\ \delta = 4,\ n= 5$, i.e. $\mu = 2$ and $S=14$. These parameters fulfill the conditions of Theorem \ref{main}. Proceeding as described in Section 4.1, we obtain $W=4$, i.e. $E=F=1$ and see that $G_1^c$ and $\overline{G}_1^c$ cannot be split. Hence, we get the condition that all non trivially zero full-size minors of the matrices 
$\begin{pmatrix} G_0\ G_1\ G_2
\end{pmatrix}$,
$G_1^c$
and
$\overline{G}_1^c$
have to be nonzero. With the help of the computer we found the following solution over 
$\mathbb{F}_{31}$:

$G_0 = \left(\begin{array}{rrrrr}
\hspace{-0.15cm} 5 \hspace{-0.15cm} & \hspace{-0.15cm} 30 \hspace{-0.15cm} & \hspace{-0.15cm} 14 \hspace{-0.15cm} & \hspace{-0.15cm} 11 \hspace{-0.15cm} & \hspace{-0.15cm} 1 \hspace{-0.15cm}\\
\hspace{-0.15cm} 3 \hspace{-0.15cm}  & \hspace{-0.15cm} 23 \hspace{-0.15cm} & \hspace{-0.15cm} 21 \hspace{-0.15cm} & \hspace{-0.15cm} 12 \hspace{-0.15cm} & \hspace{-0.15cm} 5  \hspace{-0.15cm} \end{array}\hspace{-0.15cm} \right)$,
$G_1 = \left(\begin{array}{rrrrr}
\hspace{-0.15cm} 17 \hspace{-0.15cm} & \hspace{-0.15cm} 4 \hspace{-0.15cm} & \hspace{-0.15cm} 24 \hspace{-0.15cm} & \hspace{-0.15cm} 14 \hspace{-0.15cm} & \hspace{-0.15cm} 7 \hspace{-0.15cm}\\
\hspace{-0.15cm} 7 \hspace{-0.15cm}  & \hspace{-0.15cm} 24 \hspace{-0.15cm} & \hspace{-0.15cm} 12 \hspace{-0.15cm} & \hspace{-0.15cm} 20 \hspace{-0.15cm} & \hspace{-0.15cm} 22  \hspace{-0.15cm} \end{array}\hspace{-0.15cm} \right)$ and
$G_2 = \left(\begin{array}{rrrrr}
\hspace{-0.15cm} 14 \hspace{-0.15cm} & \hspace{-0.15cm} 0 \hspace{-0.15cm} & \hspace{-0.15cm} 12 \hspace{-0.15cm} & \hspace{-0.15cm} 19 \hspace{-0.15cm} & \hspace{-0.15cm} 1 \hspace{-0.15cm}\\
\hspace{-0.15cm} 23 \hspace{-0.15cm}  & \hspace{-0.15cm} 1 \hspace{-0.15cm} & \hspace{-0.15cm} 21 \hspace{-0.15cm} & \hspace{-0.15cm} 1 \hspace{-0.15cm} & \hspace{-0.15cm} 22  \hspace{-0.15cm} \end{array}\hspace{-0.15cm} \right)$.

For the given $k$ and $\delta$, the bound of \cite{mdssr} asks for $n\geq 7$, i.e. the parameters of this example are not possible. The only paper (up to our knowledge) that presents a construction for these parameters is \cite{sm01a}, which asks for $q-1=an$ with $a\geq 5$, i.e. the smallest possible field size there is $31$ as well. However, our code has the additional advantage that for $j\in\{0,1\}$, the $j$-th column distance and the $j$-th reverse column distance are optimal.
\end{example}

In the following, we present some examples where $k\nmid\delta$.

\begin{example}
Let $k=2$, $n=3$ and $\delta=3$, i.e. $\mu=2$ and $t=1$. Note that these parameters fulfill the conditions of Theorem \ref{main2}. We determine the optimal splitting according to the procedure in Section 4.2. We obtain $D=1$ and calculate $S=6$. This gives us $E=1$, i.e. $D\geq E$. Furthermore, we calculate $F=0$. This gives us for $\ell\geq 1$, the matrices $G_1^c$ and $\tilde{G}_2$ leading to weight $3$ each, i.e. it is not possible to split $G_1^c$. For $\ell=0$, we start with the matrices $[G_0\ G_1]$ and $\tilde{G}_2$ and calculate $A_1=3\geq k$ and $A_2=t$. This means we can split $[G_0\ G_1]$ into $G_0$ and $G_1$. If all non trivially zero full-size minors of $G_1^c$ are nonzero, this already implies that all full-size minors of $G_0$ are nonzero. In summary, we get the conditions that all non trivially zero full-size minors of $G_1^c$, $G_1$ and $\tilde{G}_2$ have to be nonzero. The following example over $\mathbb F_3$ fulfills these conditions:
$$G_0=\begin{pmatrix}
    1 & 0 & 2 \\ 2 & 1 & 2
\end{pmatrix}, \ G_1=\begin{pmatrix}
    1 & 1 & 1\\ 1 & 0 & 2
\end{pmatrix}, \ G_2=\begin{pmatrix}
    1 & 1 & 1\\ 0 & 0 & 0
\end{pmatrix}.$$
For the given $k$ and $\delta$, the bound of \cite{mdssr} asks for $n\geq 6$, i.e. the parameters of this example are not possible. The only paper that presents a construction for these parameters is \cite{sm01a}, which has smallest possible field $\mathbb F_{16}$. This means we manage to improve the field size a lot and additionally, our code has optimal $j$-th column distance for $j\in\{0,1\}$.
\end{example}

To be able to also do a comparison with \cite{mdssr}, we consider the following example.

\begin{example}
Let $k=2$, $n=6$ and $\delta=3$, i.e. $\mu=2$ and $t=1$ and $S=10$. In contrast to the previous example, we can now split $G_1^c$ and just need that all full-size minors of $G_0$, $\begin{pmatrix}
    G_1\\ G_0
\end{pmatrix}$, $G_1$ and $\tilde{G}_2$ are nonzero. An example fulfilling these conditions over $\mathbb F_7$ is
$$G_0=\begin{pmatrix}
   2 & 5 & 6 & 2 & 2 & 0\\ 6 & 5 & 5 & 0 & 3 & 4
\end{pmatrix}, G_1=\begin{pmatrix}
    4 & 6 & 4 & 4 & 5 & 5\\ 1 & 4 & 0 & 2 & 5 & 2
\end{pmatrix}, G_2=\begin{pmatrix}
    1 & 1 & 1 & 1 & 1\\ 0 & 0 & 0 & 0 & 0
\end{pmatrix}.$$
The conditions from \cite{mdssr} ask all full-size minors of $[G_0\ G_1]$ and all minors of $\begin{pmatrix}
    \tilde{G}_2\\ G_1 \\ G_0
\end{pmatrix}$ to be nonzero, which obviously is a strictly stronger condition leading to larger field size (the smallest field over which we found it fulfilled is $\mathbb F_{67}$). The smallest possible field for the construction in \cite{sm01a} is also for these parameters $\mathbb F_{16}$. 

\end{example}

\section{Conclusion}
We presented new criteria for the construction of MDS convolutional codes. Moreover, the constructed convolutional codes have optimal first column distances, and if $k\mid\delta$, also optimal first reverse column distances. We also provided some examples of codes fulfilling the new criterion over smaller finite fields than in previous constructions.










\end{document}